\documentclass[a4paper]{article}
\usepackage{fullpage}
\usepackage{hyperref}
\usepackage{amsmath}
\usepackage{amssymb}
\usepackage{amsthm}

\usepackage[cref,natbib,theorems]{nikis}
\date{}
\author{Tomohiro I \and Dominik K\"{o}ppl}

\usepackage[utf8]{inputenc}
\usepackage{graphicx}
\usepackage{booktabs}

\newcommand*{\instancename}[1]{\ensuremath{\mathsf{#1}}} \newcommand*{\functionname}[1]{{\ensuremath{\renewcommand{\rmdefault}{ptm}\fontfamily{ppl}\selectfont\textrm{\textup{#1}}}}} 

\newcommand*{\fnPred}{\functionname{predecessor}}
\newcommand*{\fnDelete}{\functionname{delete}}
\newcommand*{\fnInsert}{\functionname{insert}}
\newcommand*{\fnAccess}{\functionname{access}}

\usepackage{xcolor}
\definecolor{teigiIro}{HTML}{5700B5}
\newcommand*{\teigi}[1]{{\color{teigiIro}\emph{#1}}} 

\usepackage{tikz}
\usetikzlibrary{automata}

\begin{document}
\title{Load-Balancing Succinct B Trees}
\maketitle

\begin{abstract}
	We propose a B tree representation storing $n$ keys, each of $k$ bits, 
	in either 
	(a) $nk + \Oh{nk / \lg n}$ bits or
	(b) $nk + \Oh{nk \lg \lg n/ \lg n}$ bits 
	of space 
	supporting all B tree operations in either
	(a) \Oh{\lg n } time or
	(b) \Oh{\lg n / \lg \lg n} time, respectively.
	We can augment each node with an aggregate value such as the minimum value within its subtree,
	and maintain these aggregate values within the same space and time complexities.
	Finally, we give the sparse suffix tree as an application, 
	and present a linear-time algorithm computing the sparse longest common prefix array from the suffix AVL tree of Irving et al.~[JDA'2003].
\end{abstract}

\section{Introduction}

A B tree~\cite{bayer70organization} is the most ubiquitous data structure found for relational databases
and is, like the balanced binary search tree in the pointer machine model, 
the most basic search data structure in the external memory model.
A lot of research has already been dedicated for solving various problems with B trees, and various variants of the B tree have already been proposed (cf.~\cite{graefe11btree} for a survey).
Here, we study a space-efficient variant of the B tree in the word RAM model under the context of a dynamic predecessor data structure, which provides the following methods:

\begin{description}
	\item[\fnPred($K$)] returns the predecessor of a given key~$K$ (or $K$ itself if it is already stored);
	\item[\fnInsert($K$)] inserts the key~$K$; and \item[\fnDelete($K$)] deletes the key~$K$.
\end{description}

Nowadays, when speaking about B trees we actually mean B+ trees~\cite[Sect.~3]{comer79ubiquitous} 
(also called \teigi{leaf-oriented B-tree}~\cite{bille18partial}),
where the leaves store the actual data (i.e., the keys).
We stick to this convention throughout the paper.

\subsection{Related Work}
The classic B tree as well its B+ and B$^*$ tree variants support the above methods in \Oh{\lg n} time, 
while taking $\Oh{n}$ words of space for storing $n$ keys.
Even if each key uses only $k = \oh{\lg n}$ bits, the space requirement keeps the same since its pointer-based tree topology already needs \Oh{n} pointers.
To improve the space while retaining the operational time complexity is main topic of this article.
However, this is not a novel idea:

The earliest approach we are aware of is due to \citet{blandford04compact} who proposed a representation of the leaves as blocks of size \Ot{\lg n}. 
Assuming that keys are integer of $k$ bits,
they store the keys not in their plain form, but by their differences encoded with Elias-$\gamma$ code~\cite{elias74code}.
Their search tree takes \Oh{n \lg ((2^k+n)/n)} bits while conducting B tree operations in \Oh{\lg n} time.

More recently, \citet{prezza17dynamic} presented a B tree whose leaves store between $b/2$ and $b$ keys for $b =\lg n$.
Like~\cite[Sect.~3]{bille18partial} or \cite[Thm.~6]{dietz89optimal}, 
the main aim was to provide prefix-sums by augmenting each internal node of the B tree with additional information about the leaves in its subtree such as the sum of the stored values.
Given $m$ is the sum of all stored keys plus $n$,
the provided solution uses $2n \lg (m/n) + \lg \lg n + \Oh{\lg m / \lg n}$ bits of space and supports 
B tree operations as well as prefix-sum in \Oh{\lg n} time.
This space becomes $2nk + \oh{n}$ if we store each key in plain $k$ bits.

Data structures computing prefix-sums are also important for dynamic string representations~\cite{he10succinct,navarro14dynamic,munro15compressed}.
For instance, \citet{he10succinct} use a B tree as underlying prefix-sum data structure for efficient deletions and insertions of characters into a dynamic string.
If we omit the auxiliary data structures on top of the B tree to answer prefix-sum queries, 
their B tree uses $nk + \Oh{nk / \sqrt{\lg n}}$ bits of space while supporting B tree operations in \Oh{\lg n / \lg \lg n} time, an improvement over the \Oh{\lg n} time of the data structure of \citet[Thm.~1]{gonzalez09rank} sharing the same space bound.
In the static case, \citet{delpratt07compressed} studied compression techniques for a static prefix-sum data structure.

Asides from prefix-sums, another problem is to maintain a set of strings,
where each node~$v$ is augmented with the length of the longest common prefix (LCP) among all strings stored as satellite values in the leaves of the subtree rooted at~$v$~\cite{ferragina99stringbtree}.

When all keys are distinct, the implicit dictionary of \citet{franceschini06optimal} supports 
\Oh{\lg n} time for \fnPred{} and \Oh{\lg n} amortized time for updates (\fnDelete{} and \fnInsert)
while using only constant number of words of extra space.
Allowing duplicate keys, \citet{katajainen10compact} presented a data structure with the same time bounds
but using \Oh{n \lg \lg n / \lg n} bits of extra space.

\subsection{Our Contribution}

Our contribution (cf. \cref{secSpaceEfficient}) is a combination of a generalization of the rearrangement strategy of the B$^*$ tree with the idea to enlarge the capacity of the leaves similarly to some approaches listed in the related work.
With these techniques we obtain:

\begin{theorem}\label{thmBTree}
	There is a B tree representation storing $n$ keys, each of $k$ bits, in $nk + \Oh{nk / \lg n}$ bits of space,
	supporting all B tree operations in \Oh{\lg n} time.
\end{theorem}

We stress that this representation does not compress the keys, which can be advantageous if keys are not simple data types but for instance pointers to complex data structure 
such that equality checking cannot be done by merely comparing the bit representation of the keys.
In this setting of incompressible keys, the space of a \emph{succinct} data structure supporting \fnPred{}, \fnInsert{}, and \fnDelete{} is $nk + \oh{nk}$ bits for storing $n$ keys.

We present our space-efficient B tree in \cref{secSpaceEfficient}.
Additionally, we show that we can augment our B tree with auxiliary data such that we can address the prefix-sum problem and LCP queries without worsening the query time (cf.~\cref{secAggregatedValues}).
In \cref{secAcceleration}, we show that we can speed up our B tree with a technique used by \citet{he10succinct} leading to \cref{thmBTreeFast}.
In \cref{secSparseSuffixTree}, we give the sparse suffix problem as an application, and compare our solution with the suffix AVL tree~\cite{irving03suffixavl}, 
for which
we propose an algorithm to compute the sparse LCP array.

\section{Preliminaries}
Our computational model is the word RAM model with a word size of $w$ bits.
We assume that a key uses $k = \Oh{w}$ bits, and that
we can compare two keys in \Oh{1} time. 
More precisely, we support the comparison to be more complex than just comparing the $k$-bit representation bitwise as long as it can be evaluated within constant time.
Let $n = \Oh{2^w}$ be the number of keys we store at a specific, fixed time.

\begin{figure}
	\centering{\includegraphics[scale=1.0]{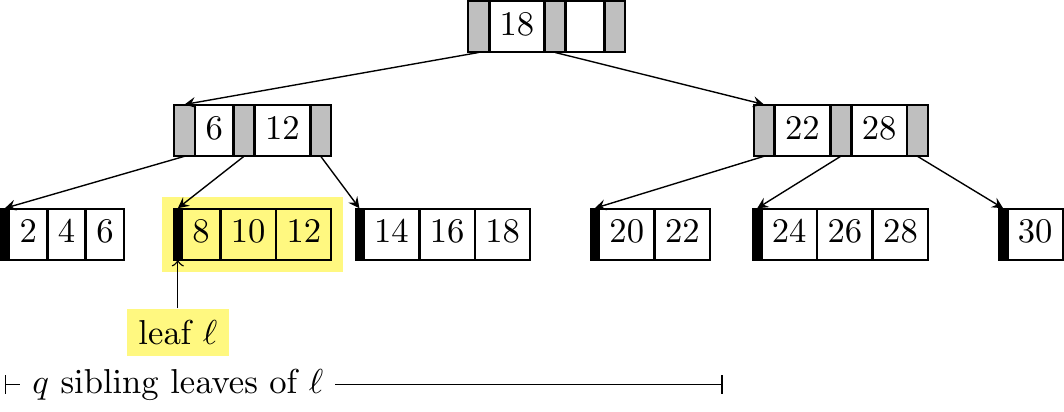}
	}\caption{A B+ tree with degree~$t=3$ and height~$3$. 
	A child pointer is a gray field in the internal nodes.
	An internal node stores $t-1$ integers where the $i$-th integer with value~$j$ regulates that only those keys of at most~$j$ go to the children in the range from the first up to the $i$-th child.
	In what follows (\cref{figInsertion}), we consider inserting the key~$9$ into the full leaf~$\ell$ (storing the keys $8$, $10$, and $12$), 
and propose a strategy different from splitting $\ell$ by considering its $q = 3$ siblings.}
	\label{figBTree}
\end{figure}

A B+ tree of degree~$t$ for a constant $t \ge 3$ is a rooted tree whose nodes have an out-degree between $\upgauss{t/2}$ and $t$.
See \cref{figBTree} for an example.
All leaves are on the same height, which is \Ot{\lg n} when storing $n$ keys.
The number of keys each leaf stores is between $\gauss{t/2}$ and $t$ (except if the root is a leaf). 
Each leaf is represented as an array of length~$t$; each entry of this array has $k$ bits. We call such an array a \teigi{leaf array}.
Each leaf additionally stores a pointer to its preceding and succeeding leaf.
Each internal node~$v$ stores an array of length~$t$ for the pointers to its children, 
and an integer array~$I_v$ of length~$t-1$ to distinguish the children for guiding a top-down navigation.
In more detail, 
$I[i]$ is a key-comparable integer such that all keys less than $I_v[i]$ are stored in the subtrees rooted at the left siblings of the $i$-th child of~$v$.
Since the integers of~$I_v$ are stored in ascending order,
to know in which subtree below~$v$ a key is stored, we can perform a binary search on~$I_v$.

A root-to-leaf navigation can be conducted in \Oh{\lg n} time, since there are \Oh{\lg n} nodes on the path from the root to any leaf, 
and selecting a child of a node can be done with a linear scan of its stored keys in $\Oh{t} = \Oh{1}$ time.

Regarding space, each leaf stores at least $t/2$ keys.
So there are at most $2n/t$ leaves.
Since a leaf array uses $kt$ bits, the leaves can use up to $2nk$ bits.
This is at most twice the space needed for storing all keys in a plain array.
In what follows, we provide a space-efficient variant:

\begin{figure}
	\centering{\includegraphics[scale=1.0]{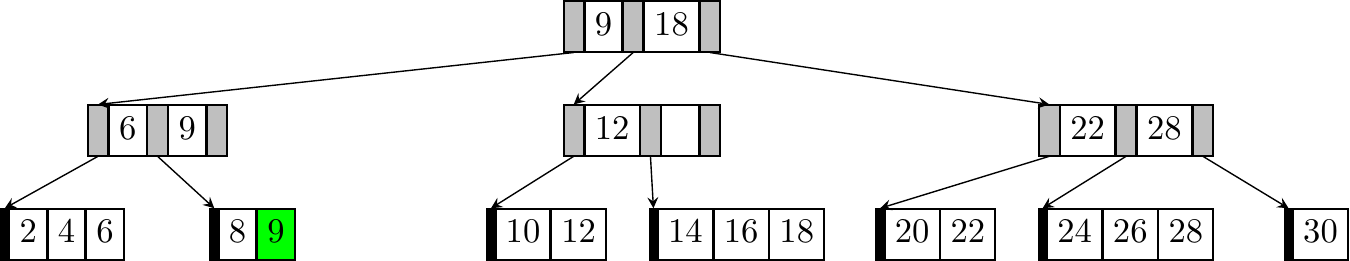}

		\vspace{1em}

		\includegraphics[scale=1.0]{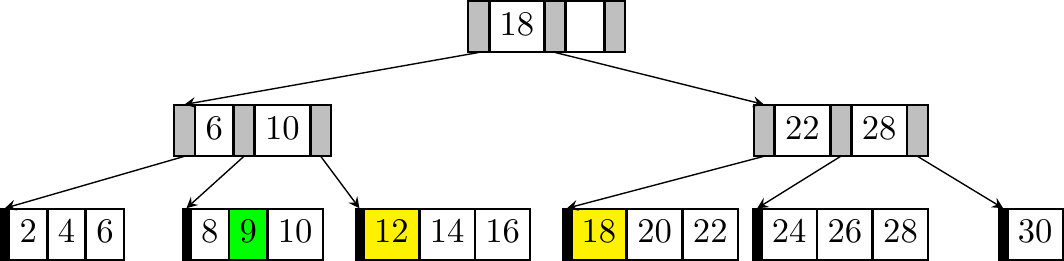}
	}\caption{\cref{figBTree} after inserting the key~9 into the leaf~$\ell$. 
		Top: The classic B+ and B$^*$ variants split $\ell$ on inserting~$9$, causing its parent to split, too.
		Bottom: In our proposed variant (cf.~\cref{secSpaceEfficient}) for $q \ge 3$, we shift the key~$12$ of $\ell$ to its succeeding leaf, 
from which we shift the key~$18$ to the next succeeding leaf, which was not yet full.}
	\label{figInsertion}
\end{figure}

\section{Space-Efficient B Trees}\label{secSpaceEfficient}
To obtain a space-efficient B tree variant, we apply several ideas.
We start with the idea to share keys among several leaves (\cref{secKeySharing}) to maintain the space of the leaves more economically.
We can adapt this technique for leaves maintaining a \emph{non-constant} number of keys efficiently (\cref{secShiftingKeys}).
However, for such large leaves we need to reconsider how and when to delete them (\cref{secManageLeaf}),
leading to the final space complexity of our proposed data structure (\cref{secFinalSpaceComplexity}) and \cref{thmBTree}.

\subsection{Key Sharing}\label{secKeySharing}
Our first idea is to keep the leaf arrays more densely filled.
For that, we generalize the idea of B$^*$ trees~\cite[Sect.~6.2.4]{knuthArt3Sorting}:
The B$^*$ tree is a variant of the B tree (more precisely, we focus on the B+ tree variant) with the aim to defer the split of a full leaf on insertion by rearranging the keys with a dedicated sibling leaf.
On inserting a key into a full leaf, we try to move a key of this leaf to its dedicated sibling.
If this sibling is also full, we split both leaves up into three leaves, 
each having $2/3 \cdot t$ keys on average~\cite[Sect.~6.2.4]{knuthArt3Sorting}.
Consequently, we have the lower bound of $2/3 \cdot t$ keys on the number of keys per leaf.
We can generalize this bound by allowing a leaf to share its keys with $q \in \Ot{\lg n}$ siblings:
Now a split of a leaf only occurs when all its $q$ siblings are already full.
After splitting these $q+1$ leaves into $q+2$ leaves,
each leaf stores $t \cdot q/(q+1)$ keys on average.
Consequently, all leaves of the tree use up to 
\begin{equation}\label{eqLeavesBits}
nk(q+1)/q = nk + nk/q = n k + \Oh{n k / \lg n}
\text{~bits for~}
q \in \Ot{\lg n}.
\end{equation}
Since each leaf stores up to $t = \Oh{1}$ keys, shifting a key to one of the $q$ siblings takes \Oh{q} time.
That is because, for shifting a key from the $i$-th leaf to the $j$-th leaf with $i < j$, we need to 
move the largest key stored in the $g$-th leaf to the $(g+1)$-th leaf for $g \in [i..j)$ (the moved key becomes the smallest key stored in the $(g+1)$-th leaf, cf.~\cref{figInsertion}).
Since a shift changes the entries of $\Oh{q}$ leaves, we have to update the information of those leaves' ancestors.
There are at most $\sum_{h=1}^{\lg n} \upgauss{q (t/2)^{-h}} = \Oh{\lg n + q}$ many such ancestors, 
and all of them can be updated in time linear to the tree height.
Thus, we obtain a B$^*$ tree variant with the same time complexities, but higher occupation rates of the leaves.

Finally, it is left to deal with deletions, which can be done symmetrically by changing the lower bound of keys a leaf stores to roughly $t \cdot q/(q+1)$ keys.
Whenever we are about to delete a key of a leaf~$\ell$ storing this minimum number of keys, 
we first try to shift a key of one of $\ell$'s $q$ sibling nodes to $\ell$.
If all these siblings also store the minimum number of keys, 
we eventually remove $\ell$ and distribute $\ell$'s keys among a constant number of $\ell$'s siblings.

\subsection{Shifting Keys Among Large Leaves}\label{secShiftingKeys}
Next, we want to slim down the entire B$^*$ tree by shrinking the number of internal nodes. 
For that, we increase the number of elements a leaf can store up to $b := w \lg n / k$.
Since a leaf now maintains a large number of keys, 
shifting a key to one of its $q$ neighboring sibling leaves takes $\Oh{bqk/w} = \Oh{\lg^2 n}$ time. 
That is because, 
for an insertion into a leaf array, 
we need to shift the stored keys to the right to make space for the key we want to insert.
We do not want to shift the keys individually since this would take \Oh{b} total time.
Instead, we can shift \Ot{w/k} keys in constant time by using word-packing, 
yielding \Oh{bk/w} time for an insertion into a leaf array.

\begin{figure}[t]
	\centering{\includegraphics[scale=1.0]{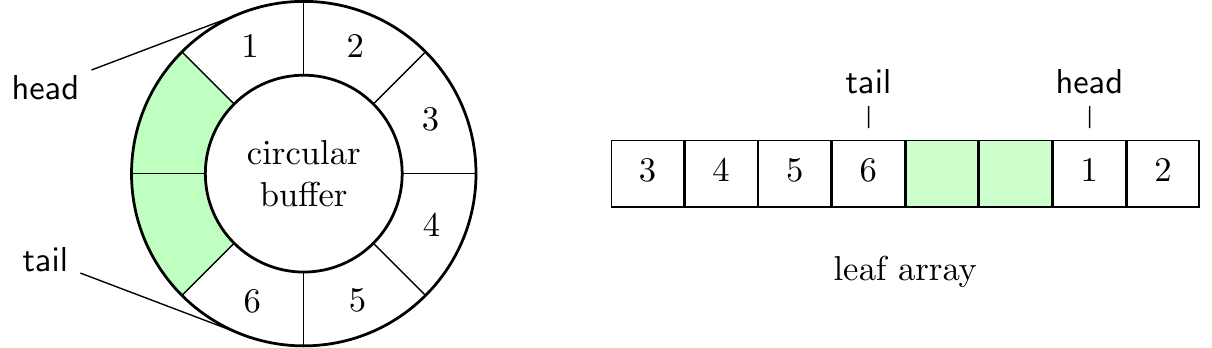}
	}\caption{A circular buffer representation of a leaf array capable of storing 8 keys.
		The pointers \textsf{head} and \textsf{tail} support prepending a key, removing the first key, appending a key, and removing the last key, all in constant time. 
	The right figure shows that the circular buffer is actually implemented as a plain array with two pointers.}
	\label{figCircularBuffer}
\end{figure}

To retain the original \Oh{\lg n} time bound, we represent the leaf arrays by \teigi{circular buffers}.
A circular buffer supports, additionally to removing or adding the last element in constant time as in a standard array,
the same operations for the first element in constant time as well. See \cref{figCircularBuffer} for a visualization.
For an insertion elsewhere, we still have to shift the keys to the right.
This can be done in $\Oh{bk/w} = \Oh{\lg n}$ time with word-packing as described above for the plain leaf array.
Finally, on inserting a key into a full leaf~$\ell$, 
we pay $\Oh{\lg n}$ time for the insertion into this full leaf,  
but subsequently can shift keys among its sibling leaves in constant time per leaf.

\subsection{Deleting Large Leaves}\label{secManageLeaf}
For supporting creating and deleting leaves efficiently, we face the following problem:
If we stick to $b \cdot q/(q+1)$ keys as the lower bound on the number of keys a leaf can store (replacing $t$ by $b$ in the bound described in \cref{secKeySharing}), 
we no longer can sustain our claimed time complexity when repetitively deleting freshly created leaves:
When deleting a leaf, we would need to move its remaining $\Ot{b}$ keys to its $q$ siblings, which would take $\Oh{bqk/w} = \Oh{n^2}$ time.
However, we would like to have a sufficiently large lower bound such that the claimed space bound of \cref{thmBTree} still holds.
Unfortunately, we have not found a way to impose such a lower bound~$m$ on each leaf without sacrificing the running time,
since a deletion of a leaf~$\ell$ storing $m$ keys costs us \Oh{m q k/w} time to move $\ell$'s keys to $\ell$'s $q$ siblings,
such that we cannot charge a leaf deletion with the gained capacity of these siblings unless $qk/w = \Oh{1}$.

As a remedy, 
we drop the lower bound on the number of keys a leaves may store, i.e., we delete a leaf only when it becomes empty,
and further impose the invariant:
\begin{equation}
		\text{Among the $q$ siblings of every non-full leaf, there is at most another non-full leaf.}
	\tag{Inv}\label{eqInv}
\end{equation}
\newcommand*{\myInvariant}{(\ref{eqInv})}
Let us first see why \myInvariant{} helps us to solve the problem regarding the space; subsequently we show how to sustain \myInvariant{} 
while retaining our operational time complexity:
By the definition of \myInvariant{}, for every $q$ subsequent leaves, there are at most two leaves that are non-full.
Consequently, these $q$ subsequent leaves store at least $qb - 2b$ keys.
Hence, the number of leaves it at most $\lambda := nq/(bq-2b)$, and
all leaves use at most 
\begin{equation}\label{eqLeavesBitsLarge}
\lambda bk =  nk qb/(qb - 2b) = nkq(q-2) = nk + 2nk/(q-2) = nk + \Oh{nk/\lg n} \text{~bits,}
\end{equation}
conforming with \cref{eqLeavesBits} for $q \in \Ot{\lg n}$.
As long as we sustain \myInvariant{}, it is impossible to delete a leaf without deleting \Om{b} keys. 
Fortunately, we can sustain \myInvariant{} by slightly changing the way we perform a deletion or an insertion of a key:
\begin{description}
	\item[Deletion] When deleting a key from a full leaf~$\ell$ having a non-full leaf~$\ell'$ as one of its $q$ siblings, we shift a key from $\ell'$ to $\ell$ such that $\ell$ is still full after the deletion.
	If $\ell'$ becomes empty, then we delete it.
	All that can be done in the same time bounds as explained for the insertion in \cref{secShiftingKeys} (shifting keys within \Oh{q} circular buffers).
\item[Insertion] Our explanation of insertions already conforms with \myInvariant{}, since we only split a leaf whenever all its $q$ siblings are full.
	In that case, we create precisely two new leaves, each inheriting half of the keys of the old leaf.
	In particular, these two leaves are the only non-full leaves among their $q$ siblings.
\end{description}

\subsection{Final Space Complexity}\label{secFinalSpaceComplexity}

Finally, we can bound the number of internal nodes by the number of leaves~$\lambda$ defined in \cref{secManageLeaf}:
Since the minimum out-degree of an internal node is $t/2$,
there are at most 
\[
	\lambda \sum_{i=1}^\infty (2/t)^k = 2\lambda/(t-2) = 
\Oh{n (q+1) /(qtb)} = \Oh{n / tb} 
\text{~internal nodes.}
\]
Since a node stores $t$ pointers to its children, it uses \Oh{t w} bits.
In total we can store the internal nodes in 
\begin{equation}\label{eqInternalBits}
	\Oh{twn/tb} = \Oh{wn/b} = \Oh{nk / \lg n} \text{~bits.}
\end{equation}
With that and \cref{eqLeavesBitsLarge}, we finally obtain \cref{thmBTree}.

\begin{figure}[t]
	\centering{\includegraphics[scale=1.0]{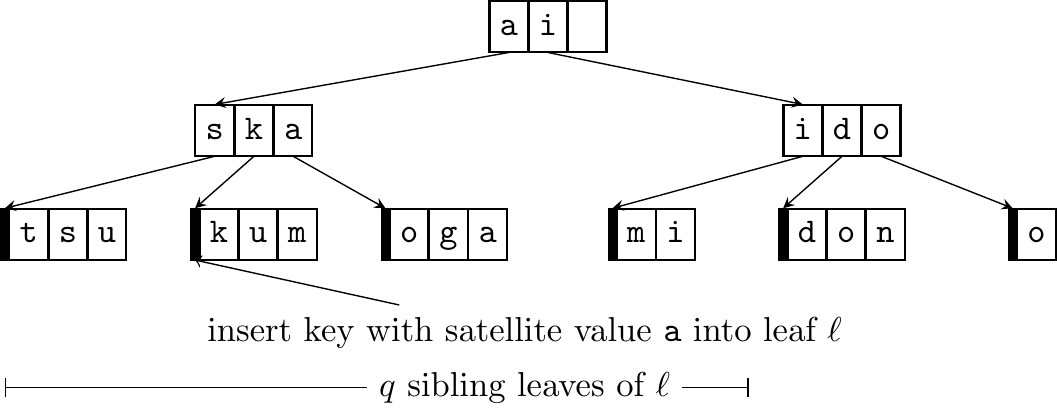}

		\vspace{0.5em}
		\tikz[baseline]{\draw[dashed] (0,.5ex)--++(5,0) ;}
		\\
		$\big\downarrow$
		inserting \texttt{a} into $\ell$
		$\big\downarrow$
		\\\noindent
		\tikz[baseline]{\draw[dashed] (0,.5ex)--++(5,0) ;}
		\vspace{1.5em}

		\includegraphics[scale=1.0]{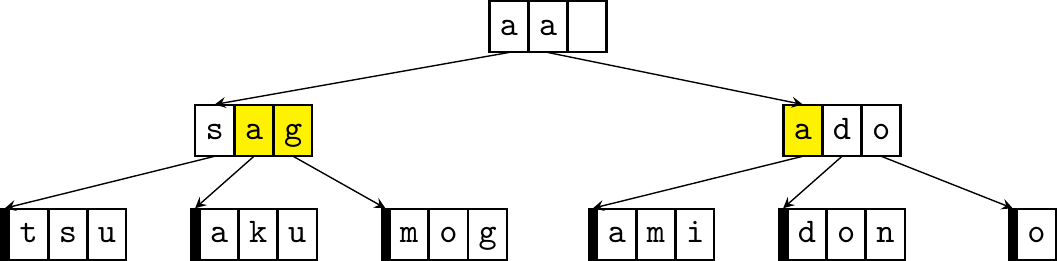}
	}\caption{Change of aggregate values on shifting keys.
	A shift causes the need to recompute the aggregate values of the satellite values stored in a leaf whose contents changed due to the shift.
	The example uses the same B tree structure as \cref{figBTree}, but depicts the satellite values (plain characters) instead of the keys. 
	Here, we used the minimum on the canonical alphabet order as aggregate function.
}
	\label{figAggregateValues}
\end{figure}

\section{Augmenting with Aggregate Values}\label{secAggregatedValues}
As highlighted in the related work, B trees are often augmented with auxiliary data
to support prefix sum queries or LCP queries when storing strings.
We present a more abstract solution covering these cases with \emph{aggregate values}, 
i.e., values composed of the satellite values stored along with the keys in the leaves.
In detail, we augment each node~$v$ with an \teigi{aggregate value} that is the return value of a decomposable aggregate function applied on the satellite values stored in the leaves of the subtree rooted at~$v$.
A \teigi{decomposable aggregate function}~\cite[Sect.~2A]{jesus15aggregation} such as the sum, the maximum, or the minimum, 
is a function~$f$ on a subset of satellite values with a constant-time merge operation~$\cdot_f$ such that,
given two disjoint subsets $X$ and $Y$ of satellite values, $f(X \cup Y) = f(X) \cdot_f f(Y)$, and
the left-hand and the right-hand side of the equation can be computed in the same time complexity.

While sustaining the methods described in the introduction like \fnPred{} for \emph{keys},
we enhance \fnInsert{} to additionally take a value as argument, and 
provide access to the aggregate values:
\begin{description}
	\item[\fnInsert($K$, $V$)] inserts the key~$K$ with satellite value~$V$;
	\item[\fnAccess($v$)] returns the aggregate value of the node~$v$; and
	\item[\fnAccess($K$)] returns the satellite value of the key~$K$.
\end{description}
To make use of $\fnAccess(v)$, the B trie also provides access to the root, and a top-down navigation based on the way $\fnPred(K)$ works, for a key~$K$ as search parameter.

For the computational analysis, let us assume that a satellite value uses \Oh{k} bits,
and that we can compute an aggregate function~$f$ bit-parallel such that it can be computed in $\Oh{bk/w} = \Oh{\lg n}$ time for a leaf storing $b = \Ot{w \lg n / k}$ values.\footnote{This assumption is of practical importance:
  Intrinsic functions like \texttt{\_mm512\_min\_epi32}, \texttt{\_mm512\_max\_epi32}, or \texttt{\_mm512\_add\_epi32}
  compute the component-wise minimum, the maximum or the summation, respectively, of two arrays with 16 integers, each of 32-bits, pairwise in one instruction. 
  (In technical terms, each of the arrays has 512 bits, and thus each stores 16 integers of 32-bits.)
  Assuming 32-bit satellite values, we can pack them into chucks of 512 bits, 
  and compute an aggregated chunk of size 512 bits in bit-parallel.
  To obtain the final aggregate value, there are again functions like
  \texttt{\_mm512\_mask\_reduce\_min\_epi32},
  \texttt{\_mm512\_mask\_reduce\_max\_epi32},
  or
  \texttt{\_mm512\_reduce\_add\_epi32} working within a single instruction.
}

Under this setting, we claim that we can obtain $\Oh{bk/w} = \Oh{\lg n}$ time for every B tree operation
while maintaining the aggregate values,
even if 
we distribute keys among $q$ leaves 
on (a) an insertion of a key into a full leaf or (b) the deletion of a key.
This is nontrivial:
For instance, when maintaining minima as aggregate values, 
if we shift the key with minimal value of a leaf~$\ell$ to its sibling, 
we have to recompute the aggregate value of~$\ell$ (cf.~\cref{figAggregateValues}), which we need to do from scratch (since we do not store additional information about finding the next minimum value).
So a shift of a key to a leaf costs $\Oh{bk/w} = \Oh{\lg n}$ time, 
resulting in $\Oh{q bk/w} = \Oh{\lg^2 n}$ overall time for an insertion.

Our idea is to decouple the satellite values from the leaf arrays where they are actually stored.
To explain this idea, let us conceptually think of the leaf arrays as a global array.
Given our B tree has $\lambda$ leaves, we partitioned this global array into $\lambda$ blocks, 
where the $i$-th block with $i \in [1..\lambda]$ starts initially at entry $1+(i-1)b$,
corresponds to the $i$-th leaf,
and has initially the same size as its corresponding leaf.
Our idea is to omit the updates of the minimum satellite value at the leaves by (a) extending or shrinking these blocks and (b) moving the block boundaries.
We no longer make the aggregate value of a leaf dependent on its leaf array, but instead let it depend on the values in its corresponding block. 
When shifting keys in the leaf arrays during an insertion or deletion of a key (cf.\ \cref{secManageLeaf}), 
we want roughly all blocks to keep the same contents by shifting the block boundaries adequately such that we only need to recompute a constant number of aggregate values stored by the leaves per update.

To track the boundaries of the blocks, 
we augment each leaf~$\ell$ with an offset value and the current size of its block.
The offset value stores the relative offset of the block with respect to the initial starting position of the block (equal to the starting position of $\ell$'s leaf array) within the global array.
We decrement the offset by one if we move a key from $\ell$ to $\ell$'s preceding sibling, 
while we increment its offset by one if we move a key of $\ell$'s preceding leaf to $\ell$.
By maintaining offset and size of the block of each leaf, we can lookup the contents of a block.
In summary, we can decouple the aggregate values with the aid of the blocks in the global array, and therefore can 
use the techniques introduced in \cref{secSpaceEfficient},
where we shift keys among $q+1$ sibling leaves, without the need to recompute the aggregate values when shifting keys.
For instance, if we shift a key from a leaf~$\ell$ to the succeeding sibling node, we still keep this key in the block of $\ell$
by incrementing its size as well as the offset of the succeeding leaf's block by one.
If we only care about insertions (and neither about deletions and blocks becoming too large) we are done since we can update $f(X)$ to $f(X \cup \{x\})$ in constant time for a new satellite value~$x \not\in X$ per definition.
However, deletions pose a problem for the running time
because we usually cannot compute $f(X \setminus \{x\})$ from $f(X)$ with $x \in X$ in constant time. 
Therefore, we have to recompute the aggregate value of a block by considering all its stored satellite values.
However, unlike leaves whose sizes are upper bounded by $b$, 
blocks can grow beyond $\om{b}$. 
The latter case makes it impossible to bound the time for recomputing the aggregate value of a block by \Oh{\lg n}.
In what follows, we show that we can retain logarithmic update time, first with a simple solution taking \Oh{\lg n} time amortized,
and subsequently with a solution taking \Oh{\lg n} worst case time.

\subsection{Updates in Batch}\label{secAggregatedBatch}
Our amortized solution takes action after a node split occurs, where it adjusts the blocks of all $q+2$ nodes that took part in that split (i.e., the full node, its $q$ full siblings and the newly created node).
The task is to evenly distribute the block sizes, 
reset the offsets, and recompute the aggregate values.
We can do all that in $\Oh{q (bk/w + \lg n)} = \Oh{\lg^2 n}$ time, since there are \Oh{q} leaves involved, 
and each leaf 
\begin{itemize}
	\item stores at most $b$ values, whose aggregate value can be computed in $\Oh{bk/w} = \Oh{\lg n}$ time, and it
	\item has \Oh{\lg n} ancestors whose aggregate values may need to be recomputed.
\end{itemize}
Although the obtained \Oh{\lg^2 n} time complexity seems costly, 
we have increased the total capacity of the $b+2$ nodes involved in the update by $\Ot{b}$ keys in total.
Consequently, before splitting one of those nodes again, we perform at least $b = \Om{\lg n}$ insertions (remember that we split a node only if it and its $q$ siblings are full).
Now, if a size of a block still becomes larger than $2b$, 
then we can afford the above rearrangement costing \Oh{\lg n} amortized time.

\begin{figure}
	\begin{tabular}{*{5}{p{0.17\textwidth}}}
		\multicolumn{1}{c}{valid} & \multicolumn{1}{c}{valid} & \multicolumn{1}{c}{invalid} & \multicolumn{1}{c}{invalid} \\
		\includegraphics[width=\linewidth]{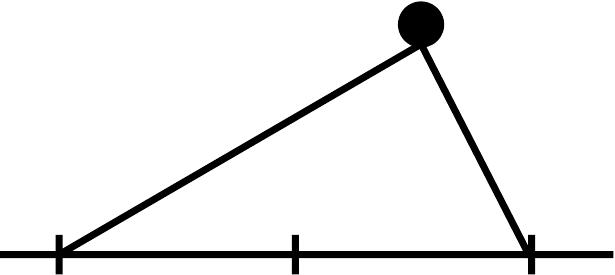} &
		\includegraphics[width=\linewidth]{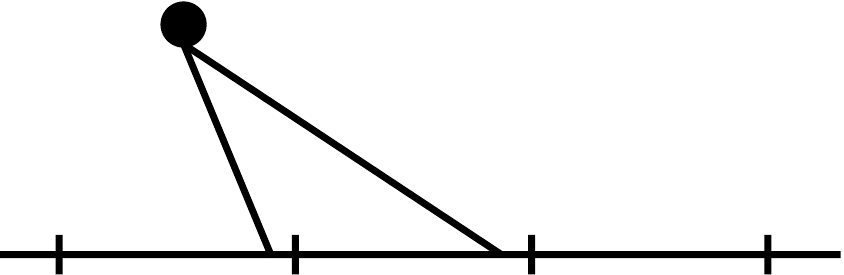} &
		\includegraphics[width=\linewidth]{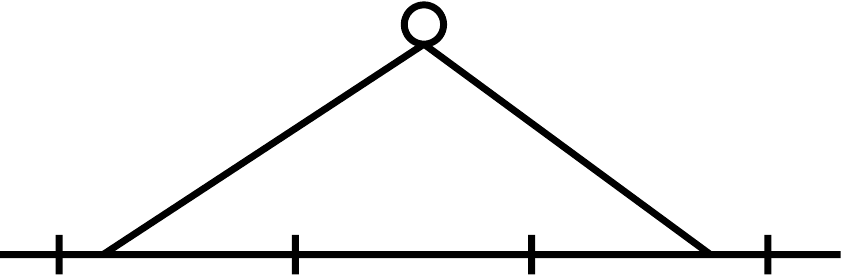} &
		\includegraphics[width=\linewidth]{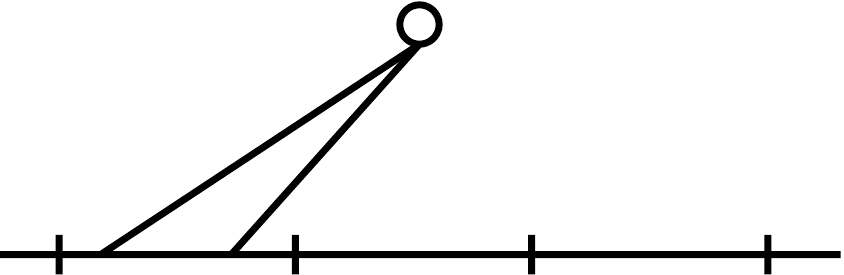} &
		\includegraphics[width=\linewidth]{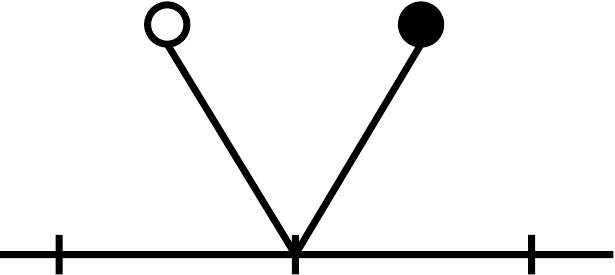} \\
	\end{tabular}
	\centering{}\caption{Valid and invalid blocks according to the definition given in \cref{secMergeUpdate}.
	The (conceptual) global array is symbolized by a horizontal line.
	The leaf arrays are intervals of the global array separated by vertical dashes.
	A dot symbolizes a leaf~$\ell$ and the intersection of the triangle spawning from~$\ell$ with the global array symbolizes the block of~$\ell$.
A node has an invalid block if its dot is hollow. The rightmost picture shows the border case that a block is invalid if its offset is $b$, while a block can be valid even if it is empty.}
\label{figInvBucket}
\end{figure}

\begin{figure}
	\centering{\includegraphics[width=\linewidth]{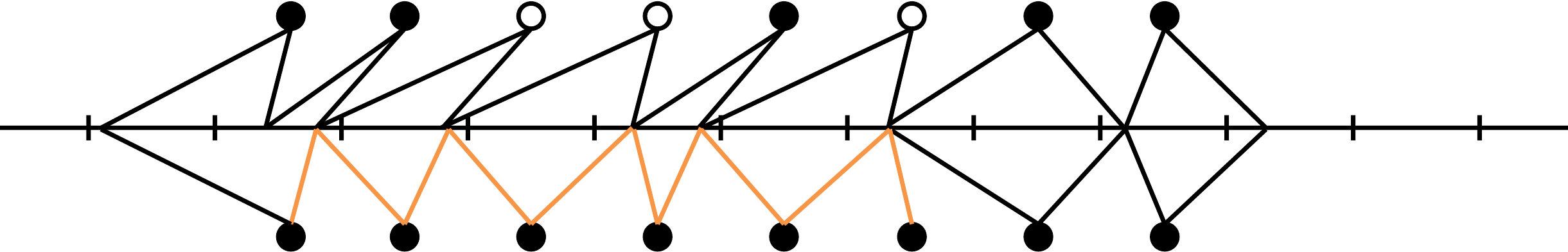}
	}\caption{Revalidation of multiple invalid blocks. The figure uses the same pictography as \cref{figInvBucket}, but additionally shows on the bottom (vertically mirrored) the outcome of our algorithm fixing the invalid blocks (\cref{secMergeUpdate}), where we empty the rightmost invalid block and swap the blocks until we find a block that can be merged with the previous block.
	}
	\label{figRevalidation}
\end{figure}

\subsection{Updates by Merging}\label{secMergeUpdate}
To improve the time bound to \Oh{\lg n} worst case time, 
our trick is to merge blocks and reassign the ownership of blocks to sibling leaves.
For the former, a merge of two blocks means that we have to combine two aggregate values,
but this can be done in constant time by the definition of the decomposable aggregate function.
To keep the size of the blocks within \Oh{b}, we watch out for blocks that become too large, which we call invalid (see \cref{figInvBucket} for a visualization).
We say a block is \teigi{valid} if
\begin{itemize}
	\item it covers at most $2b$ keys,
	\item has an offset in $(-b..b)$ (i.e., the block starts within the leaf array of the preceding leaf or of its corresponding leaf), and
	\item the sum of offset and size is at most $b$ (i.e., the block ends within the leaf array of its corresponding leaf or its succeeding leaf).
\end{itemize}
If one of those conditions becomes violated, we say that the block is \teigi{invalid}, and we take action to restore its validity.
Such a block has become invalid due to a tree update, which already costed \Oh{\lg n} time (the time for a root-to-leaf traversal).
Our goal is to rectify the invalid block~$B_i$ within the same time bound.
The first thing we do is to try to rebalance~$B_i$ with its predecessor and successor block.
To keep the following analysis simple, we only look at the preceding blocks because we can deal with the succeeding blocks symmetrically.
Now, if $B_{i-1}$ can cover $B_i$ without becoming invalid, we move all contents of $B_i$ to $B_{i-1}$ and make $B_i$ empty.
If $B_{i-1}$ would become invalid, we let $B_{i-1}$ cover the contents $B_i$ and recursively select the next preceding block $B_{i-2}$ to check whether it can merge with the old contents of $B_{i-1}$ without becoming invalid (cf.~\cref{figRevalidation}).
Since each visit of a block takes constant time (either swapping or merging contents), and we need to visit at most \Oh{q} such nodes in both directions (otherwise, all $q$ sibling nodes are full, and a split should have occurred), fixing the invalid block takes $\Oh{q} = \Oh{\lg n}$ time.

\section{Further Optimizations}
In what follows, we present features that can be applied upon the techniques introduced in the previous section.
We first start in \cref{secAcceleration} with an acceleration to \Oh{\lg n / \lg \lg n} time per B tree operation by using a sophisticated dictionary for selecting a child of a node in a top-down traversal.
Next (\cref{secCompressingKeys}), we show that we can use our B tree in conjunction with a compression of the keys stored in each leaf if the keys are plain integers indexed in the natural order.
Finally, we show that our B tree variant inherits the worst case I/O complexity of the classic B tree in the external memory model when adjusting the parameters $b, t$, and $q$.

\subsection{Acceleration with Dynamic Arrays}\label{secAcceleration}
We can accelerate the solutions of \cref{secSpaceEfficient,secAggregatedValues} by spending insignificantly more space in the lower term:

\begin{theorem}\label{thmBTreeFast}
	There is a B tree representation storing $n$ keys, each of $k$ bits, in $nk + \Oh{nk \lg \lg n/ \lg n}$ bits of space,
	supporting all B tree operations in \Oh{\lg n / \lg \lg n} time.
\end{theorem}
The idea is basically the same as 
of \citet[Lemma~1]{he10succinct}, 
who used a dynamic array data structure of \citet[Thm.~1]{raman01partialsum}, which is an application of the \emph{Q-heap} of \citet{fredman94transdichotomous}.
This dynamic array stores $\Oh{\lg^\epsilon n}$ keys in $\Oh{w \lg^{\epsilon} n}$ bits of space, and supports updates and predecessor queries, both in constant time.
It can be constructed in \Oh{\lg^\epsilon n} time, but requires a precomputed universal table of size \Oh{n^{\epsilon'}} bits
for a constant~$\epsilon' > 0$.
To make use of this dynamic array, we fix the degree~$t$ of the B tree to $t \gets \lg^\epsilon n$, 
and augment each internal node with this dynamic array to support adding a child to a node or searching a child of a node in constant time,
despite the fact that $t$ is no longer constant.
With the new degree~$t \in \Oh{\lg^\epsilon n}$, the height of our B tree becomes $\Oh{\log_t n} = \Oh{\lg n / \lg \lg n}$ 
such that we can traverse from the root to a leaf node in \Oh{\lg n / \lg \lg n} time.
Creating or removing an internal node costs $\Oh{t}$ time,
or \Oh{1} time amortized since a node stores $t/2$ to $t$ children (and hence we can charge an internal node with its $\Ot{t}$ children).

To obtain overall \Oh{\lg n / \lg \lg n} time for all leaf operations, 
we limit
\begin{itemize}
	\item the number of neighbors~$q$ to consider for node splitting or merging by setting $q \gets \lg n / \lg \lg n$, and
	\item the number of keys stored in a leaf by $b \gets w \lg n / (k \lg \lg n)$.
\end{itemize}
An insertion into a circular buffer maintaining $b$ keys can therefore be conducted in $\Oh{bk/w} = \Oh{\lg n / \lg \lg n}$ time.
Overall, adjusting $q$ and $b$ improves the running times of the B tree operations to \Oh{\lg n / \lg \lg n} time,
but increases the space of the B tree:
Now, the leaves need $nk + \Oh{nk \lg \lg n/ \lg n}$ bits according to \cref{eqLeavesBitsLarge} with $q \in \Oh{\lg n / \lg \lg n}$.
Also, the number of internal nodes increases, 
and consequently the space needed for storing the internal nodes becomes $\Oh{twn/tb} = \Oh{wn/b} = \Oh{nk \lg \lg n/ \lg n}$ bits according to \cref{eqInternalBits}.

We can accelerate also the computation in the setting of \cref{secAggregatedValues} where we maintain aggregate values:
There, we can compute the aggregate value of a leaf storing~$b = \Ot{w \lg n / (k \lg \lg n)}$ values in $\Oh{bk/w} = \Oh{\lg n / \lg \lg n}$ time.
Also, since a leaf has \Oh{\lg n / \lg \lg n} ancestors, the solution of \cref{secAggregatedBatch} conducts an update operation in \Oh{\lg n/ \lg \lg n} amortized time.
The solution of \cref{secMergeUpdate} also works in \Oh{\lg n/ \lg \lg n} time due to our choice of $b$ and $q$.

\subsection{Compressing the Keys}\label{secCompressingKeys}
In case that the keys are plain $k$-bit integers, we can store the keys with a universal coding to save space.
For that, we can follow the idea of~\cite[Section~2]{blandford04compact} and \citet{delpratt07compressed} by 
storing the keys by their differences in an encoded form such as Elias-$\gamma$ code or Elias-$\delta$ code~\cite[Sect.~V]{elias75universal}.
A leaf represents its $b' \in [1..b]$ keys by storing the first key in its plain form using $k$ bits, 
but then each subsequent key~$K_i$ by the encoded difference $K_i - K_{i-1}$, for $i \in [2..b']$.

If our trie stores the keys~$K_1, \ldots, K_n$,
then the space for the differences 
is $\sum_{i = 2}^n (2\lg (K_{i} - K_{i-1})+1) = \Oh{n \lg ((K_n+n) /n)}$ bits when using Elias-$\gamma$~\cite[Lemma~3.1]{blandford04compact} or Elias-$\delta$ code~\cite[Equation~(1)]{delpratt07compressed}.
Storing the first key of every leaf takes additional $\lambda k = nqk/(bq-2b) = \Oh{n/\lg n}$ bits.
Similar to \citet{blandford04compact},
we can replace $nk$ bits with $\Oh{n \lg ((K_n+n) /n)} = \Oh{n \lg ((2^k+n) /n)}$ bits with this compression
in our space bounds of \cref{thmBTree,thmBTreeFast},
where we implement the circular buffers as resizeable arrays.

A problem is that we applied word-packing techniques when searching a key in a leaf;
but now each leaf uses a variable amount of bits.
Under the special assumption that the word size $w$ is $\Ot{\lg n}$, i.e., the transdichotomous model~\cite{fredman94transdichotomous},
a solution is to use a lookup table~$f$ for decoding all keys stored in a bit chunk~\cite[Lemma~2]{blandford04compact}.
The lookup table~$f$ takes as input
a key~$K_i$ and a bit chunk of $(\lg n) / 2$ bits storing the keys $K_{i+1}, \cdots$ in encoded form (the first bits representing $K_{i+1}-K_{i}$).
$f$ outputs all keys that are stored in~$B$ fitting into $\lg n / 2$ bits, 
plus an $\Oh{\lg \lg n}$-bits integer storing the number of bits read from~$B$  
(in the case that the limited output space does not contain all keys stored in~$B$).
$f$ can be stored in $\Oh{k \lg \lg n \lg n 2^{(\lg n) /2}} = \oh{n}$ bits.  
$f$ can decode keys fitting into $\lg n / 2$ bits in constant time.
With~$f$, we can find the (insertion) position of a key in a circular buffer in the same time complexity as in the uncompressed version.
When inserting a key having a successor, we need to update the stored difference of this successor, but this can be done in constant time.
We can shift keys to the left and right regarding their bit positions within a circular buffer like in the uncompressed version, since we do not need to uncompress the keys.

\subsection{External Memory Model}\label{secExternalMemory}
We briefly sketch that our proposed variant inherits the virtues of the B tree in the external memory (EM) model~\cite{aggarwal88iomodel}.
For that, let us assume that we can transfer a page of $B$ words or $Bw$ bits
between the EM and the RAM of size $M > B$ words in a constant number of I/O operations (I/Os).
We assume that $Bw$ is much smaller than $n k$, otherwise we can maintain our data in a constant number of pages and need \Oh{1} I/Os for all B tree operations.
The classic B tree (and most of its variants) exhibit the property that every B tree operation can be 
processed in \Oh{\log_B n} page accesses, which is worst case optimal.
Since the EM model is orthogonal and compatible with the word RAM model,
we can improve the practical performance of the B tree by packing more keys into a single page.
We can translate our techniques to the EM model as follows:
First, we set the degree to $t \in \Ot{B}$ such that (a) an internal node fits into a constant number of pages, 
and (b) the height of our B tree is $\Ot{\log_B n}$.
Consequently, a root-to-leaf traversal costs us \Oh{\log_B n} I/Os.
If we set the number of keys a leaf can store to $b \gets wB \log_B n / k$,
then a leaf uses $wB \log_B n$ bits, or $\log_B n$ pages.
This space is maintained by a circular buffer as before, supporting insertions in \Oh{\log_B n} I/Os and
the insertion or deletion of the last or the first element in \Oh{1} I/Os.
Plugging $b =  wB \log_B n / k$ into \cref{eqLeavesBitsLarge} gives $nk + \Oh{nk / \log_B n}$ bits, which we use for storing the leaves together with the keys.
Consequently, the space of the internal nodes becomes $nk / B \log_B n$ bits (cf.~\cref{secFinalSpaceComplexity}).
For augmenting the B tree with aggregate values as explained in \cref{secAggregatedValues}, we assume that 
satellite values can be stored in \Oh{k} bits, and that an aggregate value of \Oh{1} pages of satellite values can be computed in \Oh{1} I/Os.
Unfortunately, maintaining the aggregate values naively as explained at the beginning of \cref{secAggregatedValues} comes with the cost 
of recomputing the aggregate value of a leaf, which is \Oh{\log_B n} I/Os. 
So an insertion costs us \Oh{\log_B^2 n} I/Os, 
including the cost for updating the \Oh{\log_B n} aggregate values of the ancestors of those leaves whose aggregate values have changed.
However, we can easily adapt our proposed techniques in the EM model:
For the amortized analysis \cref{secAggregatedBatch}, we can show that we perform this naive update costing \Oh{\log_B^2 n} I/Os 
after \Oh{\log_B n} operations.
Finally, the approach using merging (\cref{secMergeUpdate}) can be translated nearly literally to the EM model.
Hence, we can conclude that our space-efficient B tree variant obeys the optimal worst cast bound of \Oh{\log_B n} page accesses per operation as the classic B tree.

\section{Sparse Suffix Tree Representation}\label{secSparseSuffixTree}

\newcommand*{\stackL}     {\ensuremath{\instancename{stack}_{\textup{L}}}}
\newcommand*{\stackR}     {\ensuremath{\instancename{stack}_{\textup{R}}}}

\newcommand*{\SLCP}    {\ensuremath{\instancename{SLCP}}}
\newcommand*{\SISA}    {\ensuremath{\instancename{SISA}}}
\newcommand*{\SSA}     {\ensuremath{\instancename{SSA}}}
\newcommand*{\SAVL}     {\ensuremath{\instancename{SAVL}}}

\newcommand*{\lcp}{\ensuremath{\functionname{lcp}}}

\newcommand*{\cla}[1]{\ensuremath{\functionname{cla}_{#1}}}
\newcommand*{\cra}[1]{\ensuremath{\functionname{cra}_{#1}}}
\newcommand*{\ca}[1]{\ensuremath{\functionname{ca}_{#1}}}

\newcommand*{\migi}{\texttt{right}}
\newcommand*{\hidari}{\texttt{left}}

\newcommand*{\joui}{\ensuremath{\functionname{top}}}
\newcommand*{\node}{\ensuremath{\functionname{node}}}
\newcommand*{\kachi}{\ensuremath{\functionname{value}}}
\newcommand*{\fukusa}{\ensuremath{\functionname{depth}}}
\newcommand*{\target}{\ensuremath{\instancename{anc}}}
\newcommand*{\parent}{\ensuremath{\functionname{parent}}}

\newcommand{\ChotenRoot}[3]{\node [choten] {#1 \nodepart{lower} #2\hspace{0.5em}#3  }}
\newcommand{\Choten}[3]{node [choten] {#1 \nodepart{lower} #2\hspace{0.5em}#3  }}

Given a text~$T$, we can store starting positions of its suffixes
as keys in a B tree and use the lexicographic order of the suffixes as the order to sort the respective starting positions.
By augmenting each stored key
with the length of the LCP with the preceding key as satellite value, 
and using the minima as aggregate values stored in the internal nodes, we can represent the sparse suffix tree by our proposed B~tree.
In concrete terms, let us fix a text $T[1..n]$ of length~$n$.
Given a set of text positions~$P = \{p_1, \ldots, p_m\}$ with $p_i \in [1..n]$,
the sparse suffix array~$\SSA[1..m]$ is a ranking of $P$ with respect to the suffixes starting at the respective positions,
i.e., $T[\SSA[i]..] \prec T[\SSA[i+1]..]$ for $i \in [1..n-1]$ with $\{\SSA[1], \ldots, \SSA[m]\} = P$.
Its inverse $\SISA{}[1..n]$ is (partially) defined by $\SISA[\SSA[i]] = i$ (we only need this array conceptionally, and only care about the entries determined by this equation).
The sparse LCP array \SLCP{} stores the length of the LCP of each suffix stored in \SSA{} with its preceding entry, i.e., $\SLCP[1] = 0$ and
$\SLCP[i] := \lcp(T[\SSA[i-1]..], T[\SSA[i]..])$ for all integers~$i \in [2..|\SSA|]$.
See \cref{figExampleString} for an example of the defined arrays for $n=m$.
Finally, the \teigi{sparse suffix tree} is the compacted trie storing all suffixes~$T[p_i..]$, which can be represented by \SSA{} and \SLCP{}.

In particular, our B~tree can represent the sparse suffix tree dynamically.
By dynamic we mean $P$ to be dynamic, i.e., the support of adding or removing starting positions of suffixes to the tree while the input text is always kept static.
Dynamic sparse suffix sorting is a well-received problem that is actively investigated~\cite{i14sparse,fischer20deterministic,prezza18sparse,birenzwige20locally}.
Similarly, the suffix AVL tree~\cite{irving03suffixavl} can represent the dynamic sparse suffix tree.
The suffix AVL tree (\SAVL) is a balanced binary tree that supports dynamic operations by using a pointer-based tree topology (a formal definition follows).
Although we can retrieve \SSA{} from \SAVL{} with a simple Euler tour, 
we show in the following that retrieving \SLCP{} is also possible, but nontrivial;
this is an open problem posed in \cite[Sect.~4.7]{dissKopplDominik}.
The latter can be made trivial with a variation of \SAVL{} using more space for storing additional LCP information at each node~\cite[Sect.~4.6]{dissKopplDominik}.
We stress that, while supporting the same operations as \SAVL{}, our B~tree topology can be succinctly represented in $\oh{nk}$ bits. 
Another benefit of the B tree is that we can read \SLCP{}
from the satellite values stored in leaf arrays from left to right.
In what follows, we show that we can obtain, nevertheless, \SLCP{} from \SAVL{} by a tree traversal:

\begin{theorem}
  We can compute \SLCP{} from \SAVL{} in time linear to the number of nodes stored in \SAVL{}.
\end{theorem}

\begin{figure}
	\centering{\begin{tabular}{l*{15}{c}}
			\toprule
$i$     & 1          & 2          & 3          & 4          & 5          & 6          & 7          & 8          & 9          & 10         & 11         & 12         & 13         & 14         & 15       
			\\\midrule
$T$     & \texttt{c} & \texttt{a} & \texttt{a} & \texttt{t} & \texttt{c} & \texttt{a} & \texttt{c} & \texttt{g} & \texttt{g} & \texttt{t} & \texttt{c} & \texttt{g} & \texttt{g} & \texttt{a} & \texttt{c}  \\
$\SSA$   & 2          & 14         & 6          & 3          & 15         & 1          & 5          & 11         & 7          & 13         & 12         & 8          & 9          & 4          & 10 \\
$\SISA$   & 6          & 1          & 4          & 14         & 7          & 3          & 9          & 12         & 13         & 15         & 8          & 11         & 10         & 2          & 5 \\
$\SLCP$ & 0 & 1          & 2          & 1          & 0          & 1          & 2          & 1          & 3          & 0          & 1          & 2          & 1          & 0          & 2 \\
rules & E &A &L &L &A &D &R &A &L &A &L &L &R &D &A
			\\\bottomrule
		\end{tabular}
	}\caption{The example string $T = \texttt{caatcacggtcggac}$ used in~\cite{irving03suffixavl}. The row \emph{rules} shows from which rule or scenario (cf.~\cref{secSparseSuffixTree}) the LCP value was obtained.}
	\label{figExampleString}
\end{figure}

\subsection{The Suffix AVL Tree}
Given a set of text positions~$P$,
the \teigi{suffix AVL tree} represents each suffix~$T[p..]$ starting at a text position~$p \in P$ by a node. 
The nodes are arranged in a binary search tree topology such that reading the nodes with an in-order traversal gives the sparse suffix array.
It can take the shape of an AVL tree, which is a balanced binary tree.
For that to be performant, each node stores extra information:

Given a node~$v$ of \SAVL{},
$\cla{v}$ (resp.\ $\cra{v}$) is the lowest node having $v$ as a descendant in its left (resp.\ right) subtree.
We write $T[v..]$ for the suffix represented by the node~$v$, i.e., we identify nodes with their respective suffix starting positions.
Each node~$v$ stores a tuple~$(d_v, m_v)$, where $m_v$ is 
\begin{itemize}
   \item $\lcp(T[v..], T[\cla{v}..])$ if $d_v = \hidari$ and $\cla{v}$ exists,
   \item $\lcp(T[v..], T[\cra{v}..])$ if $d_v = \migi$ and $\cra{v}$ exists, or
   \item $0$ if $d_v = \bot$.
\end{itemize}
The value of~$d_v \in \{\hidari,\migi,\bot\}$ is set such that $m_v$ is maximized.
Let $\ca{v}$ be $\cla{v}$ (resp.\ $\cra{v}$) if $d_v = \hidari$ (resp.\ $d_v = \migi$).
If $d_v = \bot$, then $\ca{v}$ as well as $\cla{v}$ and $\cra{v}$ are not defined.
See \cref{figSAVL} for an example.

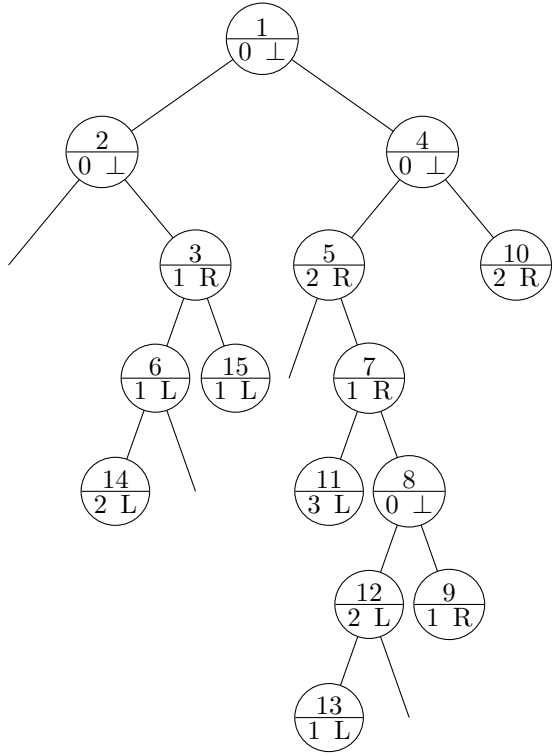
\begin{figure}[t]
	\centering{\begin{tikzpicture}[
		level 1/.style={sibling distance=12em},
		level 2/.style={sibling distance=7em}, 
		level 3/.style={sibling distance=3em}, 
	choten/.style={state with output,
		draw,
		inner sep=0.1em,
		outer sep=0, 
		}
]

\ChotenRoot{1}{0}{$\bot$}
child {\Choten{2}{0}{$\bot$} 
	child {}
	child {\Choten{3}{1}{R}
		child {\Choten{6}{1}{L}
			child {\Choten{14}{2}{L}
			}
			child {}
		}
		child {\Choten{15}{1}{L}
		}
	}
}
child {\Choten{4}{0}{$\bot$}
	child{\Choten{5}{2}{R}
		child{}
		child{\Choten{7}{1}{R}
			child {\Choten{11}{3}{L}
			}
			child{\Choten{8}{0}{$\bot$}
				child {\Choten{12}{2}{L}
					child {\Choten{13}{1}{L}
					}
					child {}
				}
				child {\Choten{9}{1}{R}
				}
			}
		}
	}
	child {\Choten{10}{2}{R}
	}
};
\end{tikzpicture}
	}\caption{The unbalanced \SAVL{} of the string
		$T = \texttt{gcaatcacggtcggac}$
		defined in \cref{figExampleString} when inserting all text positions in increasing order.
		A node consists of its position~$v$ (top), $m_v$ (bottom left) and $d_v$ (bottom right) abbreviated to L and R for \hidari{} and \migi{}, respectively.
	}
	\label{figSAVL}
\end{figure}

\subsection{Computing the Sparse LCP Array}

Since an \SAVL{} node does not necessarily store the LCP with the lexicographically preceding suffix, 
it is not obvious how to compute \SLCP{} from \SAVL{}.
For computing \SLCP{} from \SAVL{},
we use the following two facts and a lemma:
\begin{itemize}
   \item $T[\cra{v}..] \prec T[v..] \prec T[\cla{v}..]$ (in case that $\cla{v}$ and $\cra{v}$ exist)
   \item During an Euler tour (an in-order traversal), we can compute \SSA{} by reading the nodes in-order. 
      We can additionally keep track of the in-order rank~$\SISA[v]$ of each node~$v$. 
   \end{itemize}

\begin{lemma}[{\cite[Lemma 1]{irving03suffixavl}}]\label{lemLCPtransitivity}
	Given three strings~$X,Y,Z$ with the lexicographic order $X \prec Y \prec Z$, we have
	$\lcp(X,Z) = \min\left(\lcp(X,Y), \lcp(Y,Z)\right)$.
\end{lemma}

   With the following rules, we can partially compute \SLCP{}:
   \begin{description}
	 \item[\CustomLabel{RuleLeft}{Rule~L}]
	    If $d_v = \hidari$ and the right sub-tree of $v$ is empty, then $\SLCP[\SISA[v]+1] = m_v$ (since $\ca{v} = \cla{v} = \SSA[\SISA[v]+1]$ is the starting position of the lexicographically next larger suffix with respect to the suffix starting with~$v$). 

\item[\CustomLabel{RuleRight}{Rule~R}]
If $d_v = \migi$ then $\SLCP[\SISA[v]] \ge m_v$ since $v$ shares
a prefix of at least $m_v$ with the lexicographically (not necessarily next) smaller suffix $\cra{v}$.
If $v$ does not have a left child, then $\SLCP[\SISA[v]] = m_v$ since $\cra{v} = \SSA[\SISA[v]-1]$ in this case.

	 \item[\CustomLabel{RuleEmpty}{Rule~E}]
If $v$ does not have a left child and $\cra{v}$ does not exist, then $\SLCP[\SISA[v]] = 0$.
This is the case when $T[v..]$ is the smallest suffix stored in \SAVL{}.
\end{description}

To compute all \SLCP{} values, there remain two scenarios:
\begin{description}
	 \item[\CustomLabel{ScenarioDescendant}{Scenario~D}]
	If a node~$v$ has a left child, then we have to compare $v$ with the rightmost leaf in $v$'s left subtree because
		this leaf corresponds to the lexicographically preceding suffix of the suffix starting with~$v$.
	 \item[\CustomLabel{ScenarioAncestor}{Scenario~A}]
		 Otherwise, this lexicographically preceding suffix corresponds to $\cra{v}$, such that we have to compare $\cra{v}$ with $v$.
		 If $d_v = \migi$, we are already done due to~\ref{RuleRight} since $\ca{v} = \cra{v}$ in this case (such that the answer is already stored in $m_v$).
\end{description}

We cope with both scenarios by an Euler tour on \SAVL{}.
For~\ref{ScenarioDescendant}, we want to know $\lcp(T[v..], T[\cla{v}..])$ for each leaf~$v$ regardless of whether~$d_v = \hidari$ or not.
For~\ref{ScenarioAncestor}, we want to know $\lcp(T[v..], T[\cra{v}..])$ for each node~$v$ regardless of whether~$d_v = \migi$ or not.
We can obtain this lcp information by the following lemma:

\begin{lemma}\label{lemAncestorLCP}
	Given 
	$\lcp(T[\cla{v}..], T[\cra{v}..])$, 
	$\lcp(T[v..], T[\ca{v}..])$, and 
	$d_v$,
  we can compute 
  $\lcp(T[v..], T[\cla{v}..])$ and 
  $\lcp(T[v..], T[\cra{v}..])$ in constant time.
\end{lemma}
\begin{proof}
  If $d_v = \hidari$,
  \begin{itemize}
  	\item $\lcp(T[v..], T[\cla{v}..]) = \lcp(T[v..], T[\ca{v}..])$ since $\cla{v} = \ca{v}$, and
	\item $\lcp(T[v..], T[\cra{v}..]) = \lcp(T[\cla{v}..], T[\cra{v}..])$.
  \end{itemize}
  The latter is because of $T[\cra{v}..] \prec T[v..] \prec T[\cla{v}..]$ (assuming \cla{v} and \cra{v} exist) and
	\begin{align*}
		\lcp(T[\cra{v}..], T[\cla{v}..]) 
		&= \min(\lcp(T[\cra{v}..], T[v..]), \lcp(T[v..], T[\cla{v}..]))  \\
		&= \lcp(T[v..], T[\cra{v}..]) 
		\le \lcp(T[v..], T[\cla{v}..])
	\end{align*}
		according to \cref{lemLCPtransitivity}.
The case $d_v = \migi$ is symmetric:
\begin{itemize}
	\item $\lcp(T[v..], T[\cla{v}..]) = \lcp(T[\cla{v}..], T[\cra{v}..])$, and
		\item $\lcp(T[v..], T[\cra{v}..]) = \lcp(T[v..], T[\ca{v}..])$.
\end{itemize}
\end{proof}

With \cref{lemAncestorLCP}, we can keep track of $\lcp(T[\cla{v}..], T[\cra{v}..])$ while descending the tree from the root:
Suppose that we know $\lcp(T[\cla{v}..], T[\cra{v}..])$ and $v$'s left (resp.\ right) child is $x$ (resp.\ $y$).
\begin{itemize}
	\item Since $\cla{x} = v$ and $\cra{x} = \cra{v}$, $\lcp(T[\cla{x}..], T[\cra{x}..]) = \lcp(T[v..], T[\cra{v}..])$.
	\item Since $\cra{y} = v$ and $\cla{y} = \cla{v}$, $\lcp(T[\cla{y}..], T[\cra{y}..]) = \lcp(T[v..], T[\cla{v}..])$.
\end{itemize}
During an Euler tour, we keep the values $\lcp(T[\cla{u}..], T[\cra{u}..])$ in a stack for the ancestors $u$ of the current node.
By applying the above rules and using the lcp information of \cref{lemAncestorLCP} for both scenarios, we can compute the SLCP array during a single Euler tour.
This algorithm can also traverse non-balanced SAVL-trees in linear time.

\section{Conclusion}
We provided a space-efficient variation of the B tree that retains the time complexity of the standard B tree.
It achieves succinct space under the setting that the keys are incompressible.
Our main tools were the following:
First, we generalized the B$^*$ tree technique to exchange keys not only with a dedicated sibling leaf but with up to $q$ many sibling leaves.
Second, we let each leaf store \Ot{b} elements represented by a circular buffer such that moving a largest (resp.\ smallest) element of a leaf to its succeeding (resp.\ preceding) sibling can be performed in constant time.
Additionally, we could augment each node with an aggregate value and maintain these values, either with a batch update weakening the worst case time complexities to amortized time, or
with a blocking of the leaf arrays that can be maintained within the worst case time complexities.
All B tree operations can be accelerated by larger degrees~$t$ in conjunction with a smaller leaf array size~$b$ and the data structure of \citet{raman01partialsum} storing the children of an internal node,
resulting in smaller heights but more internal nodes,
which is reflected with a small increase in the lower term space complexity.
Finally, we have shown how to obtain \SLCP{} from \SAVL{} with an Euler tour storing LCP information on a stack sufficient for constructing \SLCP{} in constant time per visited \SAVL{} node.

\clearpage
\bibliographystyle{plainnat}

\begin{thebibliography}{28}
\providecommand{\natexlab}[1]{#1}
\providecommand{\url}[1]{\texttt{#1}}
\expandafter\ifx\csname urlstyle\endcsname\relax
  \providecommand{\doi}[1]{doi: #1}\else
  \providecommand{\doi}{doi: \begingroup \urlstyle{rm}\Url}\fi

\bibitem[Aggarwal and Vitter(1988)]{aggarwal88iomodel}
Alok Aggarwal and Jeffrey~Scott Vitter.
\newblock The input/output complexity of sorting and related problems.
\newblock \emph{Commun. {ACM}}, 31\penalty0 (9):\penalty0 1116--1127, 1988.

\bibitem[Bayer and McCreight(1970)]{bayer70organization}
Rudolf Bayer and Edward~M. McCreight.
\newblock Organization and maintenance of large ordered indexes.
\newblock In \emph{Proc.\ SIGFIDET}, pages 107--141, 1970.

\bibitem[Bille et~al.(2018)Bille, Christiansen, Cording, G{\o}rtz,
  Skjoldjensen, Vildh{\o}j, and Vind]{bille18partial}
Philip Bille, Anders~Roy Christiansen, Patrick~Hagge Cording, Inge~Li G{\o}rtz,
  Frederik~Rye Skjoldjensen, Hjalte~Wedel Vildh{\o}j, and S{\o}ren Vind.
\newblock Dynamic relative compression, dynamic partial sums, and substring
  concatenation.
\newblock \emph{Algorithmica}, 80\penalty0 (11):\penalty0 3207--3224, 2018.

\bibitem[Birenzwige et~al.(2020)Birenzwige, Golan, and
  Porat]{birenzwige20locally}
Or~Birenzwige, Shay Golan, and Ely Porat.
\newblock Locally consistent parsing for text indexing in small space.
\newblock In \emph{Proc.\ SODA}, pages 607--626, 2020.

\bibitem[Blandford and Blelloch(2004)]{blandford04compact}
Daniel~K. Blandford and Guy~E. Blelloch.
\newblock Compact representations of ordered sets.
\newblock In \emph{Proc.\ SODA}, pages 11--19, 2004.

\bibitem[Comer(1979)]{comer79ubiquitous}
Douglas Comer.
\newblock The ubiquitous {B}-tree.
\newblock \emph{{ACM} Comput. Surv.}, 11\penalty0 (2):\penalty0 121--137, 1979.

\bibitem[Delpratt et~al.(2007)Delpratt, Rahman, and
  Raman]{delpratt07compressed}
O'Neil Delpratt, Naila Rahman, and Rajeev Raman.
\newblock Compressed prefix sums.
\newblock In \emph{Proc.\ SOFSEM}, volume 4362 of \emph{LNCS}, pages 235--247,
  2007.

\bibitem[Dietz(1989)]{dietz89optimal}
Paul~F. Dietz.
\newblock Optimal algorithms for list indexing and subset rank.
\newblock In \emph{Proc.\ WADS}, volume 382 of \emph{LNCS}, pages 39--46, 1989.

\bibitem[Elias(1974)]{elias74code}
Peter Elias.
\newblock Efficient storage and retrieval by content and address of static
  files.
\newblock \emph{J. {ACM}}, 21\penalty0 (2):\penalty0 246--260, 1974.

\bibitem[Elias(1975)]{elias75universal}
Peter Elias.
\newblock Universal codeword sets and representations of the integers.
\newblock \emph{{IEEE} Trans. Inf. Theory}, 21\penalty0 (2):\penalty0 194--203,
  1975.

\bibitem[Ferragina and Grossi(1999)]{ferragina99stringbtree}
Paolo Ferragina and Roberto Grossi.
\newblock The string {B}-tree: {A} new data structure for string search in
  external memory and its applications.
\newblock \emph{J. {ACM}}, 46\penalty0 (2):\penalty0 236--280, 1999.

\bibitem[Fischer et~al.(2020)Fischer, I, and
  K{\"{o}}ppl]{fischer20deterministic}
Johannes Fischer, Tomohiro I, and Dominik K{\"{o}}ppl.
\newblock Deterministic sparse suffix sorting in the restore model.
\newblock \emph{ACM Trans. Algorithms}, 16\penalty0 (4):\penalty0 50:1--50:53,
  2020.

\bibitem[Franceschini and Grossi(2006)]{franceschini06optimal}
Gianni Franceschini and Roberto Grossi.
\newblock Optimal implicit dictionaries over unbounded universes.
\newblock \emph{Theory Comput. Syst.}, 39\penalty0 (2):\penalty0 321--345,
  2006.

\bibitem[Fredman and Willard(1994)]{fredman94transdichotomous}
Michael~L. Fredman and Dan~E. Willard.
\newblock Trans-dichotomous algorithms for minimum spanning trees and shortest
  paths.
\newblock \emph{J. Comput. Syst. Sci.}, 48\penalty0 (3):\penalty0 533--551,
  1994.

\bibitem[Gonz{\'{a}}lez and Navarro(2009)]{gonzalez09rank}
Rodrigo Gonz{\'{a}}lez and Gonzalo Navarro.
\newblock Rank/select on dynamic compressed sequences and applications.
\newblock \emph{Theor. Comput. Sci.}, 410\penalty0 (43):\penalty0 4414--4422,
  2009.

\bibitem[Graefe(2011)]{graefe11btree}
Goetz Graefe.
\newblock Modern {B}-tree techniques.
\newblock \emph{Foundations and Trends in Databases}, 3\penalty0 (4):\penalty0
  203--402, 2011.

\bibitem[He and Munro(2010)]{he10succinct}
Meng He and J.~Ian Munro.
\newblock Succinct representations of dynamic strings.
\newblock In \emph{Proc.\ SPIRE}, volume 6393 of \emph{LNCS}, pages 334--346,
  2010.

\bibitem[I et~al.(2014)I, K{\"{a}}rkk{\"{a}}inen, and Kempa]{i14sparse}
Tomohiro I, Juha K{\"{a}}rkk{\"{a}}inen, and Dominik Kempa.
\newblock Faster sparse suffix sorting.
\newblock In \emph{Proc.\ STACS}, volume~25 of \emph{LIPIcs}, pages 386--396,
  2014.

\bibitem[Irving and Love(2003)]{irving03suffixavl}
Robert~W. Irving and Lorna Love.
\newblock The suffix binary search tree and suffix {AVL} tree.
\newblock \emph{J. Discrete Algorithms}, 1\penalty0 (5-6):\penalty0 387--408,
  2003.

\bibitem[Jesus et~al.(2015)Jesus, Baquero, and Almeida]{jesus15aggregation}
Paulo Jesus, Carlos Baquero, and Paulo~S{\'{e}}rgio Almeida.
\newblock A survey of distributed data aggregation algorithms.
\newblock \emph{{IEEE} Commun. Surv. Tutorials}, 17\penalty0 (1):\penalty0
  381--404, 2015.

\bibitem[Katajainen and Rao(2010)]{katajainen10compact}
Jyrki Katajainen and S.~Srinivasa Rao.
\newblock A compact data structure for representing a dynamic multiset.
\newblock \emph{Inf. Process. Lett.}, 110\penalty0 (23):\penalty0 1061--1066,
  2010.

\bibitem[Knuth(1998)]{knuthArt3Sorting}
Donald~E. Knuth.
\newblock \emph{The Art of Computer Programming, Volume 3: Sorting and
  Searching}.
\newblock Addison Wesley, Redwood City, CA, USA, 1998.

\bibitem[K\"{o}ppl(2018)]{dissKopplDominik}
Dominik K\"{o}ppl.
\newblock \emph{Exploring regular structures in strings}.
\newblock PhD thesis, TU Dortmund, 2018.

\bibitem[Munro and Nekrich(2015)]{munro15compressed}
J.~Ian Munro and Yakov Nekrich.
\newblock Compressed data structures for dynamic sequences.
\newblock In \emph{Proc.\ ESA}, volume 9294 of \emph{LNCS}, pages 891--902,
  2015.

\bibitem[Navarro and Nekrich(2014)]{navarro14dynamic}
Gonzalo Navarro and Yakov Nekrich.
\newblock Optimal dynamic sequence representations.
\newblock \emph{{SIAM} J. Comput.}, 43\penalty0 (5):\penalty0 1781--1806, 2014.

\bibitem[Prezza(2017)]{prezza17dynamic}
Nicola Prezza.
\newblock A framework of dynamic data structures for string processing.
\newblock In \emph{Proc.\ SEA}, volume~75 of \emph{LIPIcs}, pages 11:1--11:15,
  2017.

\bibitem[Prezza(2018)]{prezza18sparse}
Nicola Prezza.
\newblock In-place sparse suffix sorting.
\newblock In \emph{Proc.\ SODA}, pages 1496--1508, 2018.

\bibitem[Raman et~al.(2001)Raman, Raman, and Rao]{raman01partialsum}
Rajeev Raman, Venkatesh Raman, and S.~Srinivasa Rao.
\newblock Succinct dynamic data structures.
\newblock In \emph{Proc.\ WADS}, volume 2125 of \emph{LNCS}, pages 426--437,
  2001.

\end{thebibliography}

\end{document}